\documentclass{article}
\usepackage[utf8]{inputenc}
\usepackage{amsmath}
\usepackage{amssymb}
\usepackage{amsthm}
\usepackage{fullpage}
\usepackage{verbatim}

\newtheorem{theorem}{Theorem}
\newtheorem{lemma}[theorem]{Lemma}
\newtheorem{lem}[theorem]{Lemma}

\usepackage{tikz}
\usetikzlibrary{arrows.meta}
\usetikzlibrary{patterns}
\usetikzlibrary{decorations.pathreplacing, calligraphy}

\renewcommand{\epsilon}{\varepsilon}
\usepackage[colorinlistoftodos,textsize=tiny,textwidth=2cm,color=green!50!gray]{todonotes} 

\title{A PTAS for Minimizing Weighted Flow Time on a Single Machine}
\author{Alexander Armbruster \and Lars Rohwedder \and Andreas Wiese}
\date{}

\global\long\def\N{\mathbb{N}}%

\begin{document}
\global\long\def\beg{\mathrm{beg}}%
\global\long\def\en{\mathrm{end}}%
\global\long\def\len{\mathrm{len}}%
\global\long\def\OPT{\mathrm{OPT}}%
\global\long\def\down{\downarrow}%
\global\long\def\N{\mathbb{N}}%
\global\long\def\R{\mathcal{R}}%
\global\long\def\L{\mathcal{L}}%
\global\long\def\C{\mathcal{C}}%
\global\long\def\k{k}%

\maketitle

\begin{abstract}
An important objective in scheduling literature is to minimize the sum of  weighted flow times.
We are given a set of jobs where
each job is characterized by a release time, a processing time, and
a weight. Our goal is to find a preemptive schedule on a single machine that minimizes
the sum of the weighted flow times of the jobs, where the flow time
of a job is the time between its completion time and its release time.
The currently best known polynomial time algorithm for the problem
is a $(2+\epsilon)$-approximation by Rohwedder and Wiese {[}STOC
2021{]} which builds on the prior break-through result by Batra, Garg,
and Kumar {[}FOCS 2018{]} who found the first pseudo-polynomial time
constant factor approximation algorithm for the problem, and on the
result by Feige, Kulkarni, and Li~{[}SODA 2019{]} who turned the
latter into a polynomial time algorithm. However, it remains open
whether the problem admits a PTAS. 

We answer this question in the affirmative and present a polynomial
time $(1+\epsilon)$-approximation algorithm for weighted flow time
on a single machine. 
We rely on a reduction of the problem to a geometric covering problem,
which was introduced in the mentioned $(2+\epsilon)$-approximation algorithm, losing a factor $1+\epsilon$ in the approximation ratio.
However, unlike that algorithm, we solve the resulting instances of this
problem \emph{exactly}, 
rather than losing a factor $2+\epsilon$.
Key for this is to identify and exploit
structural properties of instances of the geometric covering problem
which arise in 
the reduction from weighted flow time.
\end{abstract}

\section{Introduction}

A natural problem in scheduling is to minimize the sum of weighted
flow times of the given jobs (or equivalently, the average weighted flow time). 
In this paper, we study the setting of
a single machine in which we allow the jobs to be preempted (and resumed
later). Formally, the input consists of a set of jobs $J$, where
each job $j\in J$ is characterized by a processing time $p_{j}\in\N$,
a release time $r_{j}\in\N$, and a weight $w_{j}\in\N$. We seek
to compute a (possibly) preemptive schedule for the jobs in $J$ on a single
machine, which respects the release times, i.e., each job $j\in J$
is not processed before its release time $r_{j}$. For each job $j\in J$
we denote by $C_{j}$ the completion time of $j$ in the computed schedule,
and by $F_{j}:=C_{j}-r_{j}$ its flow time. Our objective is to minimize
$\sum_{j\in J}w_{j}F_{j}$.

It was a longstanding open problem to find a polynomial time
$O(1)$-approximation algorithm for this problem. In a break-through
result, Batra, Garg, and Kumar found a $O(1)$-approximation algorithm
in pseudo-polynomial time~\cite{Batra0K18}. Their approach is based
on a reduction of the problem to the \textsc{Demand Multicut Problem
on Trees }in which we are given a tree whose edges have capacities
and costs, and paths with demands, and we seek to satisfy the demand
of each path by selecting edges from the tree in the cheapest possible
way. The reduction to this problem loses a factor of 32, and it is
solved approximately with a dynamic program. The resulting approximation
ratio was not explicitly stated in~\cite{Batra0K18}, but it is at
least 10,000~\cite{RohwedderW21}. The algorithm was 
then 
improved to a $O(1)$-approximation algorithm in \emph{polynomial} time by
Feige, Kulkarni, and Li~\cite{feige2019polynomial}.

Subsequently, Rohwedder and Wiese improved the approximation ratio
to $2+\epsilon$~\cite{RohwedderW21}. They reduced the given instances
of weighted flow time to a geometric covering problem in which instead
of edges of the tree we select axis-parallel rectangles (which induce
a tree-structure), and instead of the paths in the tree we have vertical
downward rays whose demands we need to satisfy (see Figure~\ref{fig:geom}). This
reduction loses only a factor of $1+\epsilon$ and they solve the
resulting instances up to a factor of $2+\epsilon$ with a dynamic program  (DP). %

A question that was left open is whether weighted flow time on a single
machine admits a PTAS, which is particularly intriguing since a QPTAS has been known for more than 20 years~\cite{feige2019polynomial,chekuri2002approximation}.
In this paper we answer it affirmatively.

\subsection{Our Contribution}

We present a polynomial time $(1+\epsilon)$-approximation algorithm
for weighted flow time on a single machine. The problem is strongly
NP-hard~\cite{brucker1975complexity} and hence our approximation
ratio is best possible, unless $\mathsf{P=NP}$. We use the same reduction
as in~\cite{RohwedderW21} to a geometric covering problem. However,
we solve the resulting instance of the latter problem \emph{exactly
}in pseudo-polynomial time and hence we obtain an approximation ratio
of $1+\epsilon$ overall. Via the reduction in~\cite{feige2019polynomial}
we finally obtain an $(1+\epsilon)$-approximation in polynomial time.

The algorithm in~\cite{RohwedderW21} uses a DP in order to compute
a $(2+\epsilon)$-approximation for a given instance of the mentioned
geometric covering problem. This DP heavily uses the mentioned tree-structure
of the rectangles, but it does not use any special properties of the
rays. However, these rays and their demands are highly structured.
For each interval $[s,t]$ with $s,t\in\N$ there is a ray $L[s,t]$
in the reduced instance. This ray models that some jobs released during
$[s,t]$ need to finish after $t$, since our machine can process at
most $t-s$ units of work during $[s,t]$. In fact, the demand of
$L[s,t]$ is exactly 
the total processing time of such jobs minus $t - s$.
This induces some structural properties. For
example, for two rays $L[s,t]$, $L[s',t]$ with $s'<s$ their demands
are related: they differ by the total processing time of jobs released
during $[s,s')$ minus the term $s'-s$ (and thus this difference
is independent of $t$ and the jobs released during $[s,t]$). 

The DP in~\cite{RohwedderW21} can be thought of as a recursive algorithm
that traverses the tree-like structure of the rectangles from the
root to the leaves and selects rectangles in this order. In order
to ensure that the demand of each ray is satisfied, once we arrive
in a leaf, we would like to ``remember'' the previously selected
rectangles on the path from the leaf to the root (i.e., by including
this information in the parameters that describe each DP-cell). We
cannot afford to remember this information exactly, since then we
would have too many DP-cells. The DP in~\cite{RohwedderW21} uses
complicated techniques such as smoothing and ``forgetting'' some
rectangles (similar to~\cite{Batra0K18}) in order to bound the number
of DP-cells by a polynomial, while still remembering enough information.
Instead, we take advantage of the rays' structure and show that we
need to remember only $O_{\epsilon}(1)$ numbers (the value  $\epsilon$ comes from the reduction), which intuitively
correspond to the maximum deficits up to now from rays that intersect
with the current rectangle/node but which start further up in the
plane. In particular, in this way we do not even lose a factor of
$1+\epsilon$ in the computation of the DP but we compute an optimal
solution for this geometric covering problem in pseudo-polynomial
time. At the same time, our algorithm is
much more compact and arguably simpler 
than the DP in~\cite{RohwedderW21}.

\subsection{Other related work}

Prior to the mentioned first constant factor-approximation algorithm
in pseudo-polynomial time~\cite{Batra0K18}, Bansal and Pruhs~\cite{DBLP:journals/siamcomp/BansalP14}
found a polynomial time $O(\log\log P)$-approximation algorithm for
the general scheduling~(GSP) problem where $P$ denotes the ratio
between the largest and the smallest processing time of jobs in the
instance. GSP subsumes weighted flow time and other scheduling objectives
on a single machine. Also, for the special case of weighted flow time
where $w_{j}=1/p_{j}$ for each job $j\in J$, i.e., we minimize the
average stretch of the jobs, there is a PTAS~\cite{chekuri2002approximation,DBLP:journals/scheduling/BenderMR04}, and also for weighted flow time in the case that $P=O(1)$~\cite{chekuri2002approximation}. 

Weighted flow time is also well studied in the online setting in which
jobs become known only at their respective release times. For this
case, there is a $O(\min(\log W,\log P,\log D))$-competitive algorithm
by Azar and Touitou~\cite{azar2018improved}, where $W$ is the ratio
between the largest and the smallest weight of the jobs in the instance,
and $D$ is the ratios of the largest and smallest job densities,
being defined as $w_{j}/p_{j}$ for each job $j\in J$. This improves
and unifies several prior results, like a $O(\log W)$-competitive
algorithm and a semi-online $O(\log nP)$-competitive algorithm by
Bansal and Dhamdhere~\cite{bansal2007minimizing} and a $O(\log^{2}P)$-competitive
algorithm by Chekuri, Khanna, and Zhu~\cite{chekuri2001algorithms}.
However, there can be no online $O(1)$-competitive algorithm, as
shown by Bansal and Chan~\cite{bansal2009weighted}. On the other
hand, this is possible if the online algorithm can process jobs at a slightly faster rate (by a factor $1+\epsilon$) than the optimal solution, as Bansal and Pruhs show~\cite{DBLP:conf/stoc/BansalP03,DBLP:conf/latin/BansalP04}.

There are also FPT-$(1+\epsilon)$-approximation algorithms known
for weighted flow time on one or several machines, where the parameters
are the number of machines, an upper bound on the (integral) jobs'
processing times and weights, and $\epsilon$~\cite{wiese:LIPIcs:2018:9432}.

\section{Geometric Problem}
In~\cite{RohwedderW21}, Rohwedder and Wiese reduce the weighted
flow time problem to a two-dimensional geometric covering problem.
This reduction is similar to a reduction by Batra, Garg, and Kumar~\cite{Batra0K18}
who reduce weighted flow time to the demand multicut problem on trees.
The main difference is that the former reduction is technically more
involved, but loses only a factor of $(1+\epsilon)$ in the approximation
ratio. 
In this section,
we summarize the key properties of the reduction in~\cite{RohwedderW21}.

In order to avoid technical
difficulties later, before we describe the problem arising from the reduction, we show that we can assume that the release
times of the jobs are pairwise distinct, losing only a factor of $1+\epsilon$. In the following we write $P:=\max_{j\in J}p_{j}$.
We note that in the introduction (and some related literature) $P$ is the ratio of $\max_{j\in J}p_{j} / \min_{j\in J}p_{j}$. We justify this redefinition by the fact that 
one can transform the instance to an instance in which $p_j \in \{1,2,\dotsc,O(P)\}$ for each job $j\in J$ at a loss of factor $1 + \epsilon$ (see e.g.~\cite{RohwedderW21}), 
and then the two definitions are essentially the same.
\begin{lem} \label{lem:release-times}
By losing a factor of $1+\epsilon$ in the approximation ratio and
increasing $P$ by a factor of $n/\epsilon$, we can assume that $r_{j}\ne r_{j'}$
for any two distinct jobs $j,j'\in J$. 
\end{lem}

\begin{proof}
Assume that $J=\{1,...,n\}$. For each job $j\in J$ we define a new
release time $r'_{j}:=r_{j}+j\cdot\epsilon/n$. For these changed
release times there is a schedule of cost at most $\OPT+\sum_{j\in J}\epsilon w_{j}\le\OPT+\sum_{j\in J}\epsilon w_{j}p_{j}\le(1+\epsilon)\OPT$
which we can obtain by taking $\OPT$ and shifting it by $\epsilon$
units to the right, i.e., delaying every (partical) execution of a
job by $\epsilon$ time units. On the other hand, any schedule $S$
for the new release times $\left\{ r'_{j}\right\} _{j\in J}$ yields
a schedule for the original release dates $\left\{ r_{j}\right\} _{j\in J}$
of cost at most $S+\sum_{j\in J}\epsilon w_{j}\le S+\epsilon\OPT\le(1+\epsilon)S$
(where we write $S$ for the cost of $S$). Finally, we scale up all
processing times $\left\{ p_{j}\right\} _{j\in J}$ and release times
$\left\{ r'_{j}\right\} _{j\in J}$ by a factor $n/\epsilon$ in order
to ensure that they are integral again.
\end{proof}
We define $T:=\max_{j}r_{j}+\sum_{j}p_{j}$ and note that in the optimal
solution clearly each job finishes by time $T$ or earlier.  Also, we can assume w.l.o.g.~that $T\le nP$
by splitting the given instance if necessary.



\begin{figure}
\centering \begin{tikzpicture}[scale=1.1]
  \def\dy{0.5/3}
  
  \input{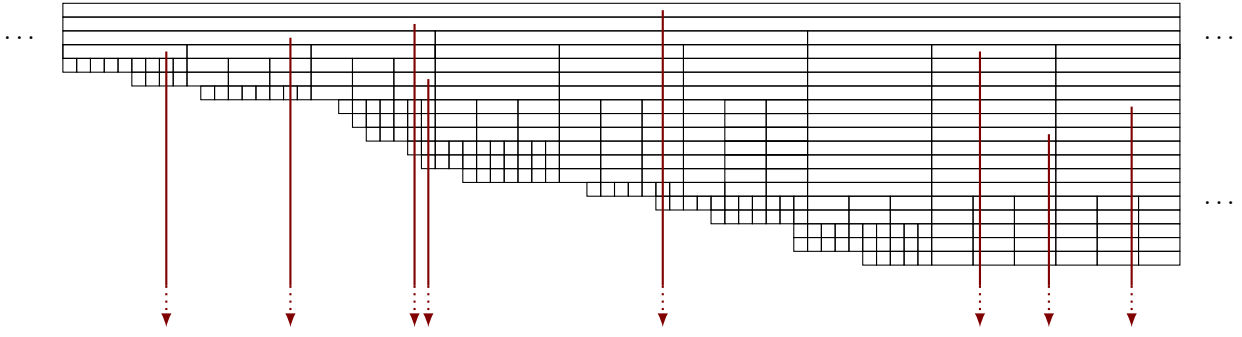}

  \draw[red!50!black, thick] (1.25, \y - 0.25) -- (1.25, 10 - 2.5 * \dy);
  \draw[-latex, red!50!black, dotted, thick] (1.25, \y - 0.25) -- (1.25, \y - 0.75);

  \draw[red!50!black, thick] (2.75, \y - 0.25) -- (2.75, 10 - 1.5 * \dy);
  \draw[-latex, red!50!black, dotted, thick] (2.75, \y - 0.25) -- (2.75, \y - 0.75);

  \draw[red!50!black, thick] (4.25, \y - 0.25) -- (4.25, 10 - 0.5 * \dy);
  \draw[-latex, red!50!black, dotted, thick] (4.25, \y - 0.25) -- (4.25, \y - 0.75);

  \draw[red!50!black, thick] (4.25 + 0.5/3, \y - 0.25) -- (4.25 + 0.5/3, 10 - 4.5 * \dy);
  \draw[-latex, red!50!black, dotted, thick] (4.25 + 0.5/3, \y - 0.25) -- (4.25 + 0.5/3, \y - 0.75);

  \draw[red!50!black, thick] (7.25, \y - 0.25) -- (7.25, 10 + 0.5 * \dy);
  \draw[-latex, red!50!black, dotted, thick] (7.25, \y - 0.25) -- (7.25, \y - 0.75);

  \draw[red!50!black, thick] (11.25 - 0.5/3, \y - 0.25) -- (11.25 - 0.5/3, 10 - 2.5 * \dy);
  \draw[-latex, red!50!black, dotted, thick] (11.25 - 0.5/3, \y - 0.25) -- (11.25 - 0.5/3, \y - 0.75);

  \draw[red!50!black, thick] (11.75 + 0.5/3, \y - 0.25) -- (11.75 + 0.5/3, 10 - 8.5 * \dy);
  \draw[-latex, red!50!black, dotted, thick] (11.75 + 0.5/3, \y - 0.25) -- (11.75 + 0.5/3, \y - 0.75);

  \draw[red!50!black, thick] (12.75 + 0.5/3, \y - 0.25) -- (12.75 + 0.5/3, 10 - 6.5 * \dy);
  \draw[-latex, red!50!black, dotted, thick] (12.75 + 0.5/3, \y - 0.25) -- (12.75 + 0.5/3, \y - 0.75);

%
%
%











\end{tikzpicture} \caption{Example of the alignment of rectangles and rays.}
\label{fig:geom} 
\end{figure}
In the reduction, we introduce for each each job $j$ a set of unit
height rectangles, and each rectangle $R$ corresponding to $j$ has
a cost $c(R)$ and a capacitiy $p(R)$. All introduced rectangles
are pairwise non-overlapping and all of their vertices have integral
coordinates, see Figure~\ref{fig:geom}. Also, they form a tree-like structure,
which we will make precise below.
For each interval of time $I = [s, t]$ we introduce a vertical (downward-pointing) ray $L(I)$ with
a demand $d(I)$. 
The goal of the geometric covering problem is to select a subset of the rectangles
of minimum cost such that the demands of all rays are covered. The
demand of a ray is covered if the total capacity of selected rectangles
intersecting with the ray is at least the demand of the ray. There
are additional technical local constraints that restrict what combinations
of rectangles can be selected. 

The rectangles and rays created by the reduction have special structural
properties that will later allow us to solve the resulting instances
optimally in pseudo-polynomial time. We will describe now these properties,
as well as the mentioned local constraints. For the complete description
of the reduction we refer to~\cite{RohwedderW21} 

\paragraph*{Hierachical grid.}
We subdivide the $x$-axis into hierachical levels of grid cells. This grid is parameterized by a constant $K = (2/\epsilon)^{1/\epsilon}$ (assuming $1/\epsilon \in \mathbb N$).
Each
cell $C$ correspond to an interval $[t_{1},t_{2}]$ with $t_{1},t_{2}\in\N$
and we define $\mathrm{beg}(C)=t_{1}$, $\mathrm{end}(C)=t_{2}$ and
$\mathrm{len}(C)=t_{2}-t_{1}$. At level $0$ there is only one cell
$C_{0}$, which contains $[0,T)$.
The precise coordinates are subject to random shifts and we do not formalize them here.
The construction ensures that $\mathrm{len}(C_{0})\le K^2 nP = O_{\epsilon}(nP)$.
The cells are organized in a tree such that each cell $C$ of some
level $\ell$ with $\len(C)>K$
has exactly $K$ children of level $\ell+1$, and these children form
a subdivision of $[\mathrm{beg}(C),\mathrm{end}(C)]$ into $K$ subintervals
of length $\mathrm{len}(C)/K$ each. A cell $C$ with $\mathrm{len}(C)\in\{1,2,\dotsc,K\}$
does not have any children. For a cell $C$ we denote its level
by $\ell(C)$. Let $\ell_{\max}$ denote the maximum level and let
$\C$ denote the set of all cells.

\paragraph*{Segments.}

Each job $j$ is associated with segments $\mathrm{Seg}(j)$ that
partition the interval $[r_{j},\mathrm{end}(C_{0})]$ (recall that
$C_{0}$ is the unique cell at level $0$). More precisely, there
is a set of cells that cover $[r_{j},\mathrm{end}(C_{0})]$, one
cell from each level $\ell\in\{0,...,\ell_{\max}\}$. Then for each
of these cells $C$ we define a set of segments $\mathrm{Seg}(j,C)$.
Each segment $S\in\mathrm{Seg}(j,C)$ is contained in $C$. Also,
if $\ell(C)\le\ell_{\max}-2$ then $S$ is a cell of level $\ell(C)+2$,
and if $\ell(C)>\ell_{\max}-2$ then $S$ is an interval $[t,t+1]$
for some $t\in\N$. In particular, the segments in $\mathrm{Seg}(j,C)$
all have the same length
and the number of segments in $\mathrm{Seg}(j,C)$ is an integral multiple of $K$ in $\{K,2K,...,K^2\}$. Furthermore, their union forms an interval
ending in $\mathrm{end}(C)$. In the case that $\ell(C) \le \ell_{\max} - 1$ this interval starts at $\mathrm{beg}(C')$ for a child $C'$ of $C$.
The set $\mathrm{Seg}(j)$ is 
the union of all segments in the sets $\mathrm{Seg}(j,C)$ for all cells $C\in \C$ and $\mathrm{Seg}(j)$ has the property
that the segment lengths are non-decreasing from left to right.

The construction of the segments guarantees the following relationship
between segments of different jobs:
Let $j,j'$ be jobs with $r_{j}\le r_{j'}$ and consider a group of segments
$\mathrm{Seg}(j', C')$. Then there exists a cell $C$ such that the interval spanned by $\mathrm{Seg}(j', C')$ is contained in the
interval spanned by $\mathrm{Seg}(j, C)$, i.e., $\bigcup_{S'\in \mathrm{Seg}(j', C')}S' \ \subseteq \bigcup_{S\in \mathrm{Seg}(j, C)}S$, and in addition $C'=C$ or $C'$ is a descendant of $C$; see also Figure~\ref{fig:segments}.

\begin{figure}
    \centering
    \begin{tikzpicture}[scale = 1.1]
  \node at (3.75, -0.4) {$\mathrm{Seg}(j)$};
  \node at (5, 0) {$r_j$};
  \draw[{|[width=4mm]}- ] (5,  -0.4) -- (5.1,  -0.4);
  \draw[|- ] (5.1,  -0.4) -- (5.2,  -0.4);
  \draw[|- ] (5.2,  -0.4) -- (5.3,  -0.4);
  \draw[{|[width=4mm]}- ] (5.3,  -0.4) -- (5.4,  -0.4);
  \draw[|- ] (5.4,  -0.4) -- (5.5,  -0.4);
  \draw[|- ] (5.5,  -0.4) -- (5.6,  -0.4);
  \draw[{|[width=4mm]}- ] (5.6,  -0.4) -- (5.9,  -0.4);
  \draw[|- ][|- ] (5.9,  -0.4) -- (6.2,  -0.4);
  \draw[|- ] (6.2,  -0.4) -- (6.5,  -0.4);
  \draw[|- ] (6.5,  -0.4) -- (6.8,  -0.4);
  \draw[|- ] (6.8,  -0.4) -- (7.1,  -0.4);
  \draw[|- ] (7.1,  -0.4) -- (7.4,  -0.4);
  \draw[|- ] (7.4,  -0.4) -- (7.7,  -0.4);
  \draw[|- ] (7.7,  -0.4) -- (8.0,  -0.4);
  \draw[|- ] (8.0,  -0.4) -- (8.3,  -0.4);
  \draw[{|[width=4mm]}- ] (8.3,  -0.4) -- (9.2,  -0.4);
  \draw[|- ] (9.2,  -0.4) -- (10.1,  -0.4);
  \draw[|- ] (10.1,  -0.4) -- (11.0,  -0.4);
  \draw[|- ] (11.0,  -0.4) -- (11.9,  -0.4);
  \draw[|- ] (11.9,  -0.4) -- (12.8,  -0.4);
  \draw[|-{|[width=4mm]} ] (12.8,  -0.4) -- (13.7,  -0.4);
  \draw[- ] (13.7,  -0.4) -- (14,  -0.4);

  \node at (3.75, -1.4) {$\mathrm{Seg}(j')$};
  \node at (6.5, -1) {$r_{j'}$};
  \draw[{|[width=4mm]}- ] (6.5,  -1.4) -- (6.6,  -1.4);
  \draw[|- ] (6.6,  -1.4) -- (6.7,  -1.4);
  \draw[|- ] (6.7,  -1.4) -- (6.8,  -1.4);
  \draw[{|[width=4mm]}- ] (6.8,  -1.4) -- (6.9,  -1.4);
  \draw[|- ] (6.9,  -1.4) -- (7.0,  -1.4);
  \draw[|- ] (7.0,  -1.4) -- (7.1,  -1.4);
  \draw[|- ] (7.1,  -1.4) -- (7.2,  -1.4);
  \draw[|- ] (7.2,  -1.4) -- (7.3, -1.4);
  \draw[|- ] (7.3,  -1.4) -- (7.4, -1.4);
  \draw[{|[width=4mm]}- ] (7.4, -1.4) -- (7.7, -1.4);
  \draw[|- ] (7.7, -1.4) -- (8.0, -1.4);
  \draw[|- ] (8.0, -1.4) -- (8.3, -1.4);
  \draw[{|[width=4mm]}- ] (8.3, -1.4) -- (9.2, -1.4);
  \draw[|- ] (9.2, -1.4) -- (10.1, -1.4);
  \draw[|- ] (10.1, -1.4) -- (11.0, -1.4);
  \draw[|- ] (11.0, -1.4) -- (11.9, -1.4);
  \draw[|- ] (11.9, -1.4) -- (12.8, -1.4);
  \draw[|-{|[width=4mm]} ] (12.8, -1.4) -- (13.7, -1.4);
  \draw[- ] (13.7,  -1.4) -- (14,  -1.4);

  \node at (3.75, -2.4) {Grid};
  \draw[- , double] (4.4, -2.4) -- (4.7, -2.4);
  \draw[{|[width=6mm]}- , double] (4.7, -2.4) -- (5, -2.4);
  \draw[|- , double] (5, -2.4) -- (5.3, -2.4);
  \draw[|- , double] (5.3, -2.4) -- (5.6, -2.4);
  \draw[{|[width=15mm]}- , double] (5.6, -2.4) -- (5.9, -2.4);
  \draw[|- , double] (5.9, -2.4) -- (6.2, -2.4);
  \draw[|- , double] (6.2, -2.4) -- (6.5, -2.4);
  \draw[{|[width=6mm]}- , double] (6.5, -2.4) -- (6.8, -2.4);
  \draw[|- , double] (6.8, -2.4) -- (7.1, -2.4);
  \draw[|- , double] (7.1, -2.4) -- (7.4, -2.4);
  \draw[{|[width=6mm]}- , double] (7.4, -2.4) -- (7.7, -2.4);
  \draw[|- , double] (7.7, -2.4) -- (8, -2.4);
  \draw[|- , double] (8, -2.4) -- (8.3, -2.4);
  \draw[{|[width=10mm]}- , double] (8.3, -2.4) -- (8.6, -2.4);
  \draw[|- , double] (8.6, -2.4) -- (8.9, -2.4);
  \draw[|- , double] (8.9, -2.4) -- (9.2, -2.4);
  \draw[{|[width=6mm]}- , double] (9.2, -2.4) -- (9.5, -2.4);
  \draw[|- , double] (9.5, -2.4) -- (9.8, -2.4);
  \draw[|- , double] (9.8, -2.4) -- (10.1, -2.4);
  \draw[{|[width=6mm]}- , double] (10.1, -2.4) -- (10.4, -2.4);
  \draw[|- , double] (10.4, -2.4) -- (10.7, -2.4);
  \draw[|- , double] (10.7, -2.4) -- (11, -2.4);
  \draw[{|[width=10mm]}- , double] (11, -2.4) -- (11.3, -2.4);
  \draw[|- , double] (11.3, -2.4) -- (11.6, -2.4);
  \draw[|- , double] (11.6, -2.4) -- (11.9, -2.4);
  \draw[{|[width=6mm]}- , double] (11.9, -2.4) -- (12.2, -2.4);
  \draw[|- , double] (12.2, -2.4) -- (12.5, -2.4);
  \draw[|- , double] (12.5, -2.4) -- (12.8, -2.4);
  \draw[{|[width=6mm]}- , double] (12.8, -2.4) -- (13.1, -2.4);
  \draw[|- , double] (13.1, -2.4) -- (13.4, -2.4);
  \draw[|- , double] (13.4, -2.4) -- (13.7, -2.4);
  \draw[{|[width=15mm]}- , double] (13.7, -2.4) -- (14.1, -2.4);
  \draw[decorate, decoration={brace, amplitude=3mm}] (13.7, -3.25) to node[below, pos=0.5, yshift=-10pt] {Cell $C$} (5.6, -3.25);
  \draw[decorate, decoration={brace, amplitude=3mm}] (8.3, 0) to node[above, pos=0.5, yshift=10pt] {$\mathrm{Seg}(j, C)$} (13.7, 0);
\end{tikzpicture}
    \caption{Alignment of segments. Small separators indicate end of a segment; big separators indicate end of a group of segments (e.g., $\mathrm{Seg}(j, C)$).
    Highlighted in this example are a particular cell $C$ and the segments $\mathrm{Seg}(j, C)$ corresponding to $j$.}
    \label{fig:segments}
\end{figure}
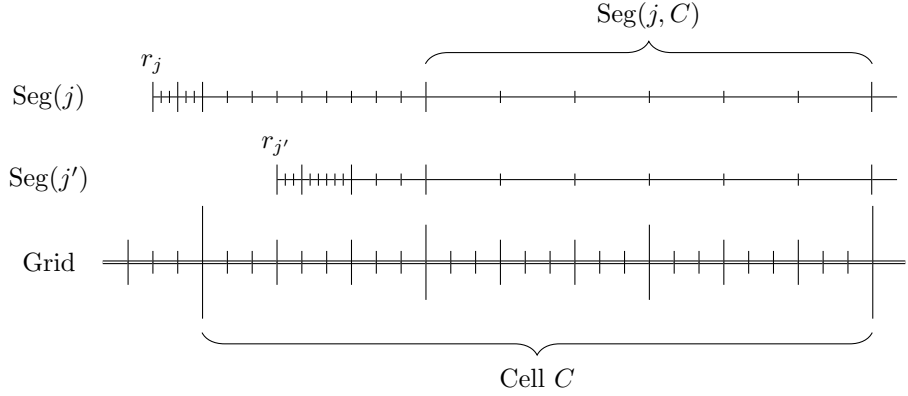

\paragraph*{Rectangles.}
We are now ready to describe the rectangles. First, assume that the
jobs correspond to integers $1,2,\dotsc,n$ and they are sorted non-decreasingly
by their release times. Consider a job $j$. For each segment $S\in\mathrm{Seg}(j)$
there is a rectangle $[\mathrm{beg}(S),\mathrm{end}(S))\times[j,j+1)$,
see Figure~\ref{fig:geom} (in our figures we assume that the origin is in the top left corner and that the $x$- and $y$-coordinates increase when going to the right and down, respectively). For each cell $C\in \C$ we define $\mathcal{R}(j,C)$ as the set of all
rectangles derived from segments in $\mathrm{Seg}(j,C)$. Also, we define $\mathcal{R}(j):=\bigcup_{C\in\C}\mathcal{R}(j,C)$
and $\mathcal{R}:=\bigcup_{j\in J}\mathcal{R}(j)$. Each rectangle
$R\in\mathcal{R}(j)$ has a cost $c(R)\in\N$ and a capacity
$p(R)$ with $p(R)=p_{j}$.

\paragraph*{Local selection constraint.}

As mentioned above, our goal is to select a subset of the rectangles
in $\R$. For this, we require a local selection constraint for the
set $\mathcal{R}(j,C)$ for each job $j\in J$ and each cell $C\in\C$
with $\mathcal{R}(j,C)\ne\emptyset$. We require that the solution
selects a prefix of the rectangles $\mathcal{R}(j,C)$, i.e., we require
that the union of the selected rectangles in $\mathcal{R}(j,C)$ forms
a rectangle that includes the leftmost rectangle in $\mathcal{R}(j,C)$.
Each such constraint is local in the sense that it affects only one
set $\mathcal{R}(j,C)$.%

\paragraph*{Rays.}
For each interval $I=[s,t]$ with $s, t \in \N$ and $0 \le s\le t \leq T$ we introduce a ray $L(I)$.
Let $j(I)$ be the job $j$ with minimum $r_{j}$ such that $s\le r_{j}$.
We define $L(I):=\{t+\frac{1}{2}\}\times[j(I)+\frac{1}{2},\infty)$,
see Figure~\ref{fig:geom} and define the demand of $L(I)$ as $d(I):=\sum_{j:s\le r_{j}\le t}p_{j}-(t-s)$. Also, denote by $\R(I)$ the rectangles that are intersected by $L(I)$.
Let $\L$ denote the set of all introduced rays.

The goal of our geometric covering problem is to select a subset $\R'\subseteq \R$ of the rectangles in $\R$ such that for each ray $L(I)\in \L$ its demand is satisfied by $\R'$, i.e., $p(\R' \cap \R(I)) \ge d(L(I))$.
In~\cite{RohwedderW21} it was shown that any pseudo-polynomial time
algorithm for this problem yields
a pseudo-polynomial time algorithm for weighted flow time, losing
only a factor of $1+\epsilon$ in the approximation ratio. Moreover,
Feige, Kulkarni, and Li~\cite{feige2019polynomial} showed that any
pseudo-polynomial time algorithm for weighted flow time can be transformed
to a polynomial time algorithm, again by losing only a factor $1+\epsilon$.
These reductions yield the following lemma. 

\begin{lemma}[\cite{feige2019polynomial}, \cite{RohwedderW21}]\label{lem:reduction}
Given a pseudo-polynomial time $c$-approximation algorithm for the
geometric covering problem defined above (as a black-box), we can
construct a polynomial time $c(1+\epsilon)$-approximation algorithm
for preemptive weighted flow time on a single machine for any constant $\epsilon>0$.
\end{lemma}

\section{Dynamic Program}

\begin{figure}
	\centering
	\begin{tikzpicture}[scale=1.1]
  \def\dy{0.5/3}

  \fill[lightgray] (4.5, 10 - 7 * \dy) rectangle (9, 10 - 8 * \dy);
  \fill[pattern=north west lines, pattern color=gray] (4.5, 10 - 7 * \dy) rectangle (9, 10 - 21 * \dy);
  \draw[|-|, thick] (4.5, 10 - 23 * \dy) to node[pos=0.5, below] {Cell $C$} (9, 10 - 23 * \dy);

  \input{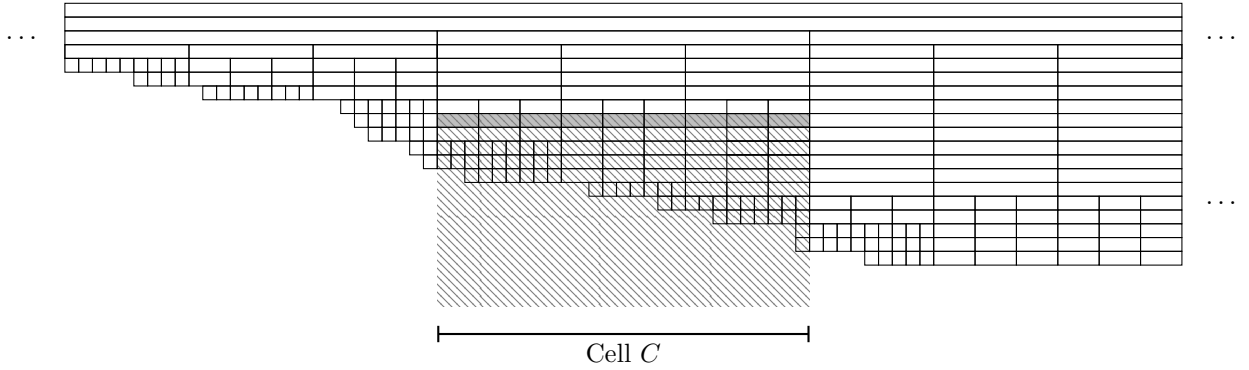}

\end{tikzpicture}
	\caption{Subproblem in the DP. Region $A$ is hatched and rectangles $\mathcal R(j, C)$ depicted in gray.}
	\label{fig:subproblem}
\end{figure}

The slightly over-simplified idea of our dynamic program is the following. We have
an entry in our DP-table (i.e., a DP-cell) for each combination of
a job $j$ and a cell $C$ such that $\mathcal{R}(j,C)\neq\emptyset$.
Let $S_{L}$ and $S_{R}$ denote the leftmost and rightmost segments
in $\mathrm{Seg}(j,C)$ respectively. Our goal is to select rectangles
contained in the region $A=[\beg(S_{L}),\en(S_{R}))\times[j,\infty)$
(see Figure~\ref{fig:subproblem}) such that we satisfy the demand of all rays that
are contained in $A$. In this DP-cell, we want to store the solution
of minimum cost with this property. In order to compute the solution
for this DP-cell, we try out all possibilities for which rectangles
in $\mathcal{R}(j,C)$ to select, obeying the rule that we need to
select a prefix of them. Then we combine this choice with the solutions
to one or more DP-cells for the next job $j+1$ (where the precise
definition of these DP-cells is induced by the tree structure in which
the rectangles are organized).

If we tried to use this approach by itself there would be a fundamental flaw:
suppose for simplicity that 
the problem is reduced to exactly one DP-cell for $(j+1,C)$ which corresponds
to the region $A'=[\beg(S_{L}),\en(S_{R}))\times[j+1,\infty)$; then the rectangles in $\mathcal{R}(j+1,C)$ are directly underneath the
rectangles in $\mathcal{R}(j,C)$, like in Figure~\ref{fig:subproblem}.
There may be rays that are contained in $A$, but not
in $A'$. In order to cover their demands, it is relevant which
rectangles in $A\setminus A'$ we have chosen.
However, so far the idea above does not take 
these rays into account
 when choosing the solution inside $A'$.

Intuitively, one approach to fix this problem is to ``remember''
the remaining demand for each ray that is contained in $A$ but not
in $A'$. In order not to forget it again in the next round of the
DP, for a given DP-cell for a job $j'$ and a cell $C'$, corresponding
to some region $A'$, we would like to remember the remaining demand
for each ray that intersects with $A'$ but that is not contained
in $A'$. By ``remember'' we mean that we include this information
in the parameters of the DP-cell. However, there is a pseudo-polynomial
number of possibilities for this remaining demand and a polynomial
number of rays in total. Thus, there is an exponential number of combinations
for these remaining demands and we cannot afford to have a DP-cell
for each of these combinations.

Instead, we show that 
we can artificially
increase the demands of certain rays \emph{contained in $A'$ }by
at most $K^2$ different amounts, 
such that any solution satisfies 
 these increased demands 
 if and only if it 
   satisfies the remaining
demands of the mentioned rays that are not contained in $A'$. It will suffice to remember
only $K^2=O_{\epsilon}(1)$ values, for which there are only a pseudo-polynomial
number of possibilities.

\subsection{Definition of the DP-table}\label{sec:DP-table}

For each rectangle $R\in\R$ we define by $\mathrm{proj}_{x}(R)$
its projection to the $x$-axis. 
Each entry in our DP-table corresponds to a specific subproblem. Note
that in addition to the description above, we introduce also DP-cells
for certain pairs of a job $j$ and a cell $C$ such that $\mathcal{R}(j,C)=\emptyset$,
for example when there is a descendant $C'$ of $C$ such that $\mathcal{R}(j,C')\ne\emptyset$.
Also, for some of our DP-cells the corresponding area $A$ can be
of the form $A=[b,\en(S_{R}))\times[j,\infty)$ for some value $b<\beg(S_{L})$,
where $S_{L}$ and $S_{R}$ denote the leftmost and rightmost segments
in $\mathrm{Seg}(j,C)$, respectively.
However, $A$ will always be contained in $C \times [j, \infty)$.
To avoid ambiguities, we will always use the term "DP-cell" for the entries of our DP-table and we will use "cell" only for the elements in $\C$.
Formally, we have a DP-cell for each combination of 
\begin{itemize}
\item a job $j\in J\cup\{n+1\}$, 
\item a cell $C\in\C$,
\item a value $\k\in\{1,...,K\}$ which intuitively indicates that we are
only interested in rectangles whose corresponding segments are contained
in the $\k$-th to $K$-th children of $C$; 
\item $j,C,$ and $\k$ induce an area $A(j, C, \k)$ defined as follows
\begin{itemize}
\item if $\ell(C)<\ell_{\max}$ then $A(j, C, \k):=[\beg(C'),\en(C))\times[j,\infty)$
where $C'$ is the $\k$-th child of $C$
\item if $\ell(C)=\ell_{\max}$ then $A(j, C, \k):=[\beg(C)+\k-1,\en(C))\times[j,\infty)$
\end{itemize}
We require that the set $A(j, C, \k)$ is nonempty, otherwise the DP-cell is not defined. Note that this happens if and only if $\ell(C)= \ell_{\max}$ and $\k > \en(C)-\beg(C)$.
\item a value $f_{C'}\in\{0,...,\sum_{j': j'<j}p_{j'}\}$ for each cell $C'\in\C(C,\k)$
such that 
\begin{itemize}
\item if $\ell(C)\le\ell_{\max}-2$ then $\C(C,\k)$ contains all grandchildren
of $C$ that are children of the $\k$-th to $K$-th children of $C$
\item if $\ell(C)=\ell_{\max}-1$ then $\C(C,\k)$ contains each interval
$[t,t+1)$ with $t\in\N$ such that $[t,t+1)$ is contained in the
$\k'$-th child of $C$ for some $\k'\in\{\k,...,K\}$ (one may thing
of $[t,t+1)$ being a cell of some dummy level $\ell_{\max}+1$), 
\item if $\ell(C)=\ell_{\max}$ then $\C(C,\k)$ contains each interval
$[t,t+1)$ with $t\in\N$ and $[t,t+1)\subseteq[\beg(C)+\k,\en(C))$. 
\end{itemize}
\item we require for each job $j'$ and each cell $C'$ that either all
rectangles in $\R(j',C')$ are contained in $A(j, C, \k)$ and that in this
case $C'$ is a descendant of $C$ or $C=C'$, or that none
of the rectangles in $\R(j',C')$ intersect with $A(j, C, \k)$ (otherwise
the DP-cell $(j,C,\k,\{f_{C'}\}_{C'})$ is not defined)
\end{itemize}
Fix such a DP-cell $(j,C,\k,\{f_{C'}\}_{C'})$. We define $\mathcal{R}_{\down}(j,C,\k):=\{R\in\mathcal{R} \mid R\subseteq A(j, C, \k)\}$,
see Figure~\ref{fig:subproblem}. The subproblem corresponding to $(j,C,\k,\{f_{C'}\}_{C'})$
is to select a set of rectangles $\R'\subseteq\R$ such that 
\begin{itemize}
\item $\R'\subseteq\mathcal{R}_{\down}(j,C,\k)$, 
\item for each job $j'\in J$ and each cell $C'\in\C$ the set $\R'\cap\mathcal{R}(j',C')$
forms a prefix of $\mathcal{R}(j',C')$, 
\item for each ray $L(I)\in\L$ for some interval $I$ with $L(I)\subseteq A(j, C, \k)$
its demand $d(I)$ is covered by $\R'$, i.e., $p(\R'\cap\R(I))\ge d(I)$, 
\item for each cell $C'\in\C(C,\k)$ and each ray $L([r_{j},t]) \in \L$ with $t\in C'$
(and for which hence $L(I)\subseteq A(j, C, \k)$ holds) we have even that $p(\R'\cap\R([r_{j},t]))\ge d([r_{j},t])+f_{C'}$;
for the case $j=n+1$ we define $r_{n+1}:=T+1$.
\end{itemize}
The objective is to minimize $c(\R')$. We store in the DP-cell $(j,C,\k,\{f_{C'}\}_{C'})$
the optimal solution to this sub-problem, or the information that there
is no feasible solution for it. In the former case, we denote by $\OPT(j,C,\k,\{f_{C'}\}_{C'})$
the optimal solution to the DP-cell.

\subsection{Filling in the DP-table}
We describe now how to fill in the entries of the DP-table. 
Consider
a DP-cell $(j,C,\k,\{f_{C'}\}_{C'})$ and assume that it has a feasible solution. Let $A:=A(j, C, \k)$. The base case arises when $A$ does not contain any rectangles from $\R$ (this holds for example when $j=n+1$). In this case
we store the empty solution in $(j,C,\k,\{f_{C'}\}_{C'})$. This is
justified due to the following lemma since $A$ can contain only rays $L([s,t])$ with $\R([s,t])=\emptyset$ and the last condition of the definitions of the sub-problem applies to such a ray only if $s=r_j'$ for some job $j'$.
\begin{lem}\label{lem:empty-ray}
		Let  $L([s, t])\in\L$ be a ray with $\R([s, t])= \emptyset$. Then $d([s,t])\leq 0$ and $s$ is not the release time of any job.
	\end{lem}
\begin{proof}
	Recall that the demand satisfies $d([s,t])= \sum_{i:s \leq r_i \leq t}p_i-(t-s)$. Suppose towards contradiction that there exists a job $j$ with $s\leq r_j \leq t$. Because we have a ray for $[s,t]$ it follows that $t \leq \en(C_0)$. The segments corresponding to $j$ partition the interval $[r_j, \en(C_0)]$, so there is a rectangle $R$ belonging to $j$ with $t \in \mathrm{proj}_{x}(R)$. But this implies $R \in \R(L)$ contradicting $\R(L)= \emptyset$. We conclude that there is no job $j$ with $s\leq r_j \leq t$. If $s$ is the release time of a job $j$ then $s\leq r_j \leq t$, so $s$ is cannot be a release time of any job. Furthermore $d([s,t])= \sum_{i:s \leq r_i \leq t}p_i-(t-s)=0-(t-s) \leq 0$.
\end{proof}
Suppose for the remainder of this section that $A$ contains at least one rectangle from $\R$. 
In our next case will handle what we will call canonical DP-cells.
We say that $(j,C,\k,\{f_{C'}\}_{C'})$ is a \emph{canonical DP-cell} if it satisfies the following properties 
\begin{itemize}
\item $\R(j,C)\ne\emptyset$,
\item for each rectangle $R\in\R(j,C)$ we have that $R\subseteq A$, and
\item the leftmost $x$-coordinate of $A$ coincides with the leftmost $x$-coordinate
of the leftmost rectangle in $\R(j,C)$. 
\end{itemize}
Note that due to this specification also the
rightmost $x$-coordinate of $A$ coincides with the rightmost $x$-coordinate
of the rightmost rectangle in $\R(j,C)$. 
Suppose that $(j,C,\k,\{f_{C'}\}_{C'})$ is a canonical DP-cell.
Let $\R^{*}$
be the optimal solution to $(j,C,\k,\{f_{C'}\}_{C'})$ and define
$\mathcal{R}^{*}(j,C):=\R^{*}\cap\mathcal{R}(j,C)$. Next,
we define values
$\{f_{C'}^{*}\}_{C'\in\C(C,\k)}$: for each $C'\in\C(C,\k)$ there
is one rectangle $R_{C'}\in\mathcal{R}(j,C)$ with $\mathrm{proj}_{x}(R_{C'})=C'$
and the value $f_{C'}^{*}$ will depend on whether we select
$R_{C'}$ in the optimal solution or not. We set 
\[
f_{C'}^{*}:=\max\{ 0,\ f_{C'}+p_{j}-r_{j+1}+r_{j}-p(\R^{*}\cap\{R_{C'}\}) \} .
\]

\begin{figure}
	\centering
	\begin{tikzpicture}[scale=1.75]
  \def\dy{0.5/4}

  \fill[lightgray] (5, 10 - 11 * \dy) rectangle (9, 10 - 12 * \dy);
  \draw[|-|, thick] (1, 10 - 21 * \dy) to node[pos=0.5, below] {Cell $C$} (9, 10 - 21 * \dy);
 
  \fill[pattern=north west lines, pattern color=gray] (5, 10 - 12 * \dy) rectangle (9, 10 - 19 * \dy);

  \input{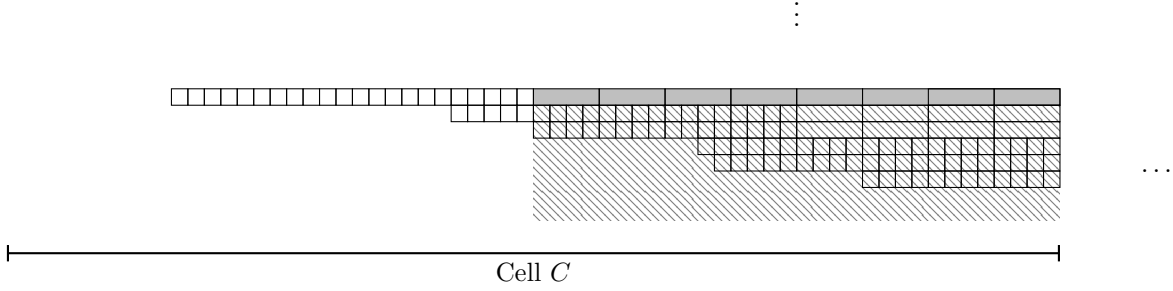}

\end{tikzpicture}
	\caption{Canonical DP-cell. The area $A$ (hatched) spans the same range on the $x$-axis as $\mathcal R(j, C)$ (gray).}
	\label{fig:my_label}
\end{figure}
\begin{lem}
\label{lem:DP-canonical}Assume that $(j,C,\k,\{f_{C'}\}_{C'})$ is
a canonical DP-cell. Then the DP-cell $(j+1,C,\k,\{f_{C'}^{*}\}_{C'})$
exists, it has a feasible solution, and $\mathcal{R}^{*}(j,C)\cup\OPT(j+1,C,\k,\{f_{C'}^{*}\}_{C'})$
is an optimal solution for the DP-cell $(j,C,\k,\{f_{C'}\}_{C'})$. 
\end{lem}
\begin{proof}
	Note that a DP-cell can only be canonical if $j\leq n$, because $\R(j, C) \neq \emptyset$. First we need to show the existence of the DP-cell $(j+1,C,\k,\{f_{C'}^{*}\}_{C'})$. 
	The only nontrivial part about the conditions for existence (see Section~\ref{sec:DP-table}) are the upper bound on $f_{C'}$ and the last property. Let $C' \in \C(C, \k)$. As $f_{C'}^{*} \leq f_{C'}+p_{j}$ and $f_{C'} \leq \sum_{j': j'<j}p_{j'}$ we obtain $f_{C'}^{*}\leq \sum_{j': j'<j+1}p_{j'}$, so $f_{C'}^{*}$ is within the given bounds.
	As $A(j, C, \k) \setminus A(j+1, C, k) \subseteq \mathbb{R} \times [j, j+1)$, a rectangle $R \in  \R(j', C')$ for a job $j'$ and a cell $C'$ with $ R \subseteq A(j, C, \k)$ is contained in $A(j+1, C, k)$ if and only if $j'\geq j+1$. 
	This implies that a rectangle $R \in \R(j', C')$ is contained in $A(j+1, C, \k)$ if and only if $j'\geq j+1$ and $R \subseteq A(j, C, k)$.
	As $(j,C,\k,\{f_{C'}\}_{C'})$ is a DP-cell, either all rectangles in $\R(j', C')$ are contained in $A(j, C, \k)$ or none of them intersect $A(j, C, \k)$. This implies that either all rectangles in $\R(j', C')$ are contained in $A(j+1, C, \k)$ or none of them intersect $A(j+1, C, \k)$. It follows that the DP-cell $(j+1,C,\k,\{f_{C'}^{*}\}_{C'})$ exists.
		
	Next we will show that $\R'':=\R^{*} \cap \mathcal{R}_{\down}(j+1,C,\k)$ is a feasible solution for $(j+1,C,\k,\{f_{C'}^{*}\}_{C'})$. 
	The first three properties of $\R''$ being a feasible solution
	(see Section~\ref{sec:DP-table}) follow directly from the fact that $\R^{*}$ is feasible for $(j,C,\k,\{f_{C'}\}_{C'})$. 
	For the fourth property, let $C'\in \C(C,\k)$, $t \in C'$ and suppose that there is a ray $L([r_{j+1}, t])$ as there is nothing to show if the ray does not exists. Then also the ray $L([r_{j}, t])$ exists.
	As $\R^{*}$ is feasible for $(j,C,\k,\{f_{C'}\}_{C'})$ we know that the ray $L([r_{j}, t])$ is covered by $\R^{*}$, i.e. $p(\R^{*}\cap\R([r_{j},t]))\ge d([r_{j},t])+f_{C'}$. As $\R([r_{j+1}, t])\cup \{R_{C'}\} = \R([r_{j}, t])$, it follows that
	\begin{align*}
		p(\R'' \cap \R([r_{j+1}, t]))&= p(\R^{*} \cap \R([r_{j+1}, t]))\\
		&=p(\R^{*} \cap \R([r_{j}, t]))-p(\R^{*} \cap \{R_{C'}\})\\
		&\geq d([r_{j},t])+f_{C'}-p(\R^{*} \cap \{R_{C'}\})\\
		&=d([r_{j+1},t])+f_{C'}+p_j-r_{j+1}+r_j-p(\R^{*} \cap \{R_{C'}\}) .
	\end{align*}
	As all release times are different by Lemma~\ref{lem:release-times} it follows $\L([r_{j+1}, t]) \subseteq A(j+1, C, \k)$. Hence, the third property yields $p(\R'' \cap \R([r_{j+1}, t]))\geq d([r_{j+1},t])$, which together with the inequality above implies that
	\begin{align*}
	    p(\R'' \cap \R([r_{j+1}, t])) &\geq d([r_{j+1},t])+\max\{0,\ f_{C'}+p_j-r_{j+1}+r_j-p(\R^{*} \cap \{R_{C'}\})\} \\
	    &= d([r_{j+1},t])+f_{C'}^{*} .
	\end{align*}
	Therefore, $\R''$ is a feasible solution for $(j+1,C,d,\{f_{C'}^{*}\}_{C'})$.
		
	Now we will prove that $\R':=\mathcal{R}^{*}(j,C)\cup \OPT(j+1,C,\k,\{f_{C'}^{*}\}_{C'})$ is feasible for $(j,C,\k,\{f_{C'}\}_{C'})$. Clearly $\R' \subseteq \mathcal{R}_{\down}(j,C,\k)$. 
	As $\mathcal{R}^{*}(j,C)$ is a prefix of $\mathcal{R}(j,C)$ and we only select a prefix of $\R(j', C')$ for any $j',C'$ in the solution for $(j+1,C,\k,\{f_{C'}^{*}\}_{C'})$, the set $\R'$ also contains only a prefix of any set $\R(j', C')$.
		
	It remains to check that the demands are covered sufficiently. 
	As the DP-cell  $(j,C,\k,\{f_{C'}\}_{C'})$ is
	a canonical DP-cell, all rectangles in $\mathcal{R}_{\down}(j,C,\k)$ corresponding to $j$ are in $\R(j, C)$. 
	So it follows that $\mathcal{R}_{\down}(j,C,\k)=\mathcal{R}_{\down}(j+1,C,\k)\dot{\cup} \R(j, C)$.
		
	Consider an interval $[s,t]$ with $L([s, t])\subseteq A(j, C, \k)$.
	In the case that $\R([s, t]) \cap \mathcal{R}_{\down}(j+1,C,\k)= \emptyset$, it follows from $\R^{*}$ being a solution for $(j,C,\k,\{f_{C'}\}_{C'})$ that the demand $d([s, t])$ is covered by $\R^{*}$, so it is covered by $\R^{*} \cap \R([s, t]) \subseteq \R^{*} \cap \R(j, C) = \R^{*}(j, C) \subseteq \R'$. If $s=r_j$ and $t \in C'$ for some $ C' \in \C(C, \k) $ then the same argument shows that $\R'$ covers the extended demand $d([r_j, t])+f_{C'}$. Note that the case $L([s, t])\subseteq A(j, C, \k)$ always occurs if $j=n$ because $\mathcal{R}_{\down}(n+1,C,\k)= \emptyset$. Thus we we will assume for the remaining cases that $j<n$.
		
	In the case that $\R([s, t]) \subseteq \mathcal{R}_{\down}(j+1,C,d)$, it follows from $\OPT(j+1,C,\k,\{f_{C'}^{*}\}_{C'})$ being a solution for $(j+1,C,\k,\{f_{C'}^{*}\}_{C'})$ that $\R'$ covers the demand of $[s,t]$. 
		
	So the remaining case is that $\R([s, t])$ contains rectangles in $\mathcal{R}_{\down}(j+1,C,d)$ and rectangles in $\R(j, C)$. 
	The first part implies $t \geq r_{j+1}$ because the rectangles in $\mathcal{R}_{\down}(j+1,C,d)$ all have a left $x$-coordinate of at least $r_{j+1}$. The second part together with $L([s,t]) \subseteq A(j, C, \k)$ implies $L([s,t])=[t+\frac{1}{2}] \times [j+\frac{1}{2}, \infty)$.
	So we get $L([s,t])=L([r_j, t])$ and therefore $\R([s,t])=\R([r_j, t])$. As there is no job released between $s$ and $r_j$ it follows that
	\begin{equation*}
	    d([s,t])=\sum_{i:s\le r_{i}\le t}p_{i}-(t-s)=\sum_{i:r_j\le r_{i}\le t}p_{i}-(t-r_j)-(s-r_j)=d([r_j, t])-(s-r_j)\leq d([r_j, t]) .
	\end{equation*}
	So if we cover the demand of $[r_j, t]$ we also cover the demand of $[s,t]$. For this reason we only show that the demands with $s=r_j$ are covered.
		
	As $t \geq r_{j+1}$ there is a ray $L([r_{j+1}, t])$ and since all release times are different by Lemma~\ref{lem:release-times} it follows  $L([r_{j+1}, t])\subseteq A(j+1, C,d)$. 
	For this ray it holds that \begin{equation*}
	    p(\OPT(j+1,C,\k,\{f_{C'}^{*}\}_{C'}) \cap \R([r_{j+1}, t])) \geq d([r_{j+1}, t])+f_{C'}^{*} .
	\end{equation*} 
	The ray $L([r_{j+1}, t])$ intersects with a rectangle $R \in \R(j, C)$, if and only if $t \in \mathrm{proj}_{x}(R)$. 
	There is a unique $C' \in \C(C, \k)$ with $t \in C'$ and this $C'$ is the projection of a rectangle $R_{C'} \in \R(j, C)$. 
	Thus, $t \in \mathrm{proj}_{x}(R)$ is equivalent to $R=R_{C'}$. This implies
	\begin{align*}
		p(\R' \cap\R([r_{j}, t]))&=p(\OPT(j+1,C,d,\{f_{C'}^{*}\}_{C'}) \cap \R([r_{j}, t])) + p(\R^{*}(j, C)\cap \R([r_{j}, t])) \\
		&=p(\OPT(j+1,C,d,\{f_{C'}^{*}\}_{C'}) \cap \R([r_{j+1}, t])) + p(\R^{*}(j, C)\cap \{R_{C'}\})\\
		&\geq d([r_{j+1}, t])+f_{C'}^{*}+p(\R^{*}(j, C)\cap \{R_{C'}\})\\
		&= d([r_{j}, t])-p_j+r_{j+1}-r_j+f_{C'}^{*}+p(\R^{*}(j, C)\cap \{R_{C'}\})\\
		&\geq d([r_{j}, t])+f_{C'} .
	\end{align*}
	Here we used in the second-to-last step that $d([r_{j}, t])=d([r_{j+1}, t])+p_j-r_{j+1}+r_{j}$, which follows from the definition of the demand $d([r_j, t])= \sum_{i:r_j \leq r_i\leq t}p_i-(t-r_j)$ and the fact that all release times are different. As $f_{C'} \geq 0$ we also get $p(\R' \cap\R([r_{j}, t]))\geq d([r_j, t])$. Altogether this implies that $\R'$ is a feasible solution.
		
	Since $c(\R')=c(\R^{*}(j, C))+c(\OPT(j+1,C,\k,\{f_{C'}^{*}\}_{C'}))\leq c(\R^{*}(j, C)) +c(\R'')=c(\R^{*})$ and $\R^{*}$ is optimal, we conclude that $\R'$ is an optimal solution.
\end{proof}
Suppose now that $(j,C,\k,\{f_{C'}\}_{C'})$ is a non-canonical DP-cell (i.e., not a canonical cell) and that $\ell(C)<\ell_{\max}$.
Let $C^{(k)}$ be the $\k$-th child of $C$. For each $C''\in\C(C^{(k)},1)$
we define a value $g_{C''}^{*}$ as follows: we identify the (unique)
cell $C'\in\C(C,\k)$ with $C''\subseteq C'$ and we define $g_{C''}^{*}:=f_{C'}.$
Also, we define $\left\{ f_{C'}^{*}\right\} _{C'}$ to be the restriction
of $\left\{ f_{C'}\right\} _{C'}$ to the intervals in $\C(C,\k+1)$.

\begin{figure}
	\centering
	\begin{tikzpicture}[scale=1.75]
  \def\dy{0.5/4}

  \fill[lightgray] (5, 10 - 11 * \dy) rectangle (9, 10 - 12 * \dy);
  \draw[|-|, thick] (1, 10 - 21 * \dy) to node[pos=0.6, below] {Cell $C$} (9, 10 - 21 * \dy);
  \draw[|-|, thick] (3, 10 - 23 * \dy) to node[pos=0.5, below] {Cell $C^{(k)}$} (5, 10 - 23 * \dy);

  \fill[pattern=north west lines, pattern color=gray] (5.03, 10 - 11 * \dy) rectangle (9, 10 - 19 * \dy);
  \fill[pattern=north east lines, pattern color=gray] (3, 10 - 11 * \dy) rectangle (4.97, 10 - 19 * \dy);

  \input{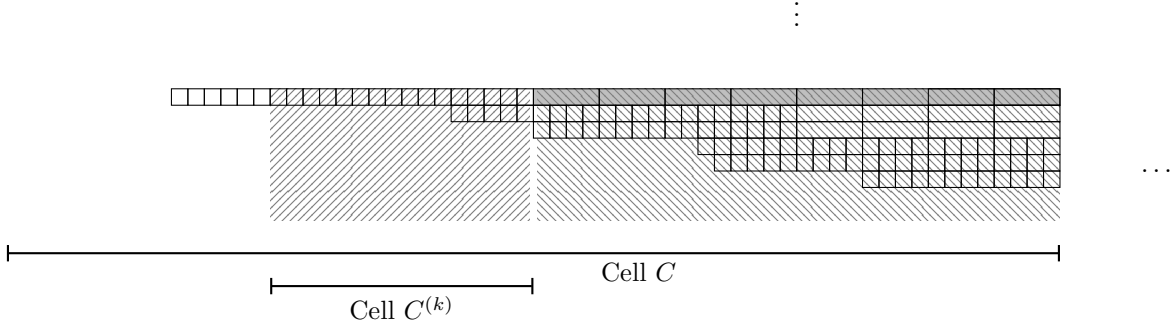}

\end{tikzpicture}
	\caption{A non-canonical DP-cell for the cell $C$. Its corresponding area $A$ is the union of the two hatched areas. 
 Its optimal solution decomposes into two independent solutions for other DP-cells (the two hatched areas). 
 In the figure we have that $k<K$. If $k=K$ 
 then we would reduce only to the left DP-cell.}
	\label{fig:my_label}
\end{figure}
\begin{lem}
\label{lem:DP-non-canonical-1}Assume that $(j,C,\k,\{f_{C'}\}_{C'})$
is a non-canonical DP-cell and that $\ell(C)<\ell_{\max}$. Then $(j,C^{(k)},1,\{g_{C''}^{*}\}_{C''})$
is a DP-cell and if $k<K$ then also $(j,C,\k+1,\{f_{C'}^{*}\}_{C'})$
is a DP-cell. 
\begin{itemize}
\item If $k<K$ then $\OPT(j,C^{(k)},1,\{g_{C''}^{*}\}_{C''})\cup\OPT(j,C,\k+1,\{f_{C'}^{*}\}_{C'})$
is an optimal solution for $(j,C,\k,\{f_{C'}\}_{C'})$. 
\item If $k=K$ then $\OPT(j,C^{(k)},1,\{g_{C''}^{*}\}_{C''})$ is an optimal
solution for $(j,C,\k,\{f_{C'}\}_{C'})$. 
\end{itemize}
\end{lem}
\begin{proof}
We start by showing that the union of the projections of $\R(j, C)$ to the $x$-axis forms an interval $I = [\mathrm{end}(C^{(i)}), \mathrm{end}(C))$, where $C^{(i)}$ is the $i$-th child of $C$ and $i \ge k$.
		The projections of the rectangles $\R(j, C)$ are the segments $\mathrm{Seg}(j, C)$ and the union of these segments
		is by definition a possibly empty interval $I$ ending in $\en(C)$. 
        This interval starts with the leftmost $x$-coordinate of the leftmost rectangle in $\R(j, C)$ if $\R(j, C) \neq \emptyset$ and is empty otherwise. 
        Let $s$ be the left endpoint of $I$ if $\R(j, C) \neq \emptyset$ and let $s= \en(C)$ if $\R(j, C) = \emptyset$. 
        Then $I=[s, \en(C))$. 
        We need to show that $s$ is the endpoint of a child cell of $C$. 
        If $\R(j, C) = \emptyset$ then as $\en(C)=\en(C^{(K)})$ we get $I=[\en(C^{(K)}), \en(C))$. 
		So for now consider $\R(j, C) \neq \emptyset$. 
        As $\ell(C) \leq \ell_{\max}-1$ the structure of the segments implies that there is a value $i' \in [K]$ such that $s=\beg(C^{(i')})$ where $C^{(i')}$ is the $i'$-th child of $C$. 
        As $i'<k$ would imply that some rectangles from $\R(j, C)$ are contained in $A(j, C, \k)$ and others are not, this contradicts the existence of the DP-cell $(j,C,\k,\{f_{C'}\}_{C'})$. 
        So $i' \geq k$. As $i'=k$ implies that the DP-cell $(j,C,\k,\{f_{C'}\}_{C'})$ is a canonical DP-cell, this is also not possible. 
        So $i' \geq k+1$ . So with substituting $i=i'-1$ and using $i'\geq 2$ we know that the $i$-th child of $C$ exists and as $\en(C^{(i'-1)})= \beg(C^{(i')})=s$ we get $I=[\en(C^{(i)}), \en(C))$. 

		We show now the existence of the DP-cell $(j,C^{(k)},1,\{g_{C''}^{*}\}_{C''})$. 
  Since the values $\{g_{C''}^{*}\}_{C''}$ are taken from the values $\{f_{C'}\}_{C'}$ the upper bounds for them are automatically respected. Thus,
        it suffices to show that the last property of the definition of the DP-cells (see Section~\ref{sec:DP-table}) holds.
        Let $C'$ be a cell such that $\R(j, C') \neq \emptyset$. 
        Suppose that at least one rectangle from $\R(j, C')$ intersects with $A(j, C^{(k)}, 1)$. 
        Then also at least one rectangle from $\R(j, C')$ intersects with $A(j, C, \k)$ because $A(j, C^{(k)}, 1) \subseteq A(j, C, \k)$. 
        As the DP-cell $(j,C,\k,\{f_{C'}\}_{C'})$ exists, every rectangle from $\R(j, C')$ is contained in $A(j, C, \k)$ and $C'=C$ or $C'$ is a descendant of $C$. 
        As $I=[\en(C^{(i)}), \en(C))$ with $i \geq k$ and $\en(C^{(k)}) \not\in C^{(k)}$ we obtain $I \cap C^{(k)}= \emptyset$. This implies $C' \neq C$. So $C'$ is a descendant of $C$. 
        Then $C' \subseteq C''$ for a child $C''$ of $C$. 
        If $C'' \neq C^{(k)}$, then $C' \cap C^{(k)} \subseteq C'' \cap C^{(k)}= \emptyset$ which contradicts that a rectangle in $\R(j, C')$ intersects with $A(j, C^{(k)}, 1)$. 
        So we get $C''=C^{(k)}$ and therefore $C'=C^{(k)}$ or $C'$ is a descendant of $C^{(k)}$. 
        So the last property holds for $j$ and each cell $C'$.
		
		Next, consider a job $j'\ne j$ and a cell $C'$ with $\R(j', C') \neq \emptyset$.
		Suppose again that at least one rectangle from $\R(j', C')$ intersects with $A(j, C^{(k)}, 1)$. This implies $j' \geq j$ and thus $j'>j$.
		From the definition of the segments we know that there exists a cell $C''$ such that the union of $\mathrm{Seg}(j', C')$ is contained in the union of $\mathrm{Seg}(j, C'')$ and additionally it holds that either $C'=C''$ or $C'$ is a descendant of $C''$. 
        Then at least one rectangle from $\R(j, C'')$ intersects with $A(j, C^{(k)}, 1)$. 
        So the part for $j'=j$ above implies that every rectangle in $\R(j, C'')$ is contained in $A(j, C^{(k)}, 1)$ and $C''=C^{(k)}$ or $C''$ is a descendant of $C^{(k)}$. 
        Together we get that every rectangle in $\R(j', C')$ is contained in $A(j, C^{(k)}, 1)$ and that $C'=C^{(k)}$ or $C'$ is a descendant of $C^{(k)}$. 
        This proves the last property of the definition of our DP-cells and we conclude that the DP-cell $(j,C^{(k)},1,\{g_{C''}^{*}\}_{C''})$ exists.
		
		Now we consider the case $k=K$ (we will  consider the more complicated case $k<K$ afterwards).
		We show that the sub-problems $(j,C,\k,\{f_{C'}\}_{C'})$ and $(j,C^{(k)},1,\{g_{C''}^{*}\}_{C''})$ have the same set of solutions. 
		By definition of $A$ it holds that $A(j, C, K)=A(j, C^{(K)}, 1)$, so $\R_{\downarrow}(j, C, K)=\R_{\downarrow}(j, C^{(K)}, 1)$.
	Consider a ray $L(I)$ for some interval $I$. Then $L(I) \subseteq A(j, C, K)$ is equivalent to $L(I) \subseteq A(j, C^{(K)},1)$, so both sub-problems need to cover the same demands. It remains to verify that the values of $\{f_{C'}\}_{C'}$ and $\{g^*_{C''}\}_{C''}$ yield equivalent conditions for the additional demands for some of the rays.
        The union of the intervals in $\C(C, K)$ is $C^{(K)}$ which is the same as the union of the intervals in $\C(C^{(K)}, 1)$.
        Let $t \in C^{(k)}$ with $t \geq r_j$. Then there exists $C'' \in \C(C^{(K)}, 1)$ with $t \in C''$ and $C' \in \C(C, K)$ with $C'' \subseteq C'$. Note that $f_{C'}=g_{C''}^{*}$.
        This implies $d([r_j, t])+f_{C'}=d([r_j, t])+g_{C''}^{*}$, so the extended demands for each ray $L([r_j, t])$ with $t \in C^{(k)}$ are also the same. 
        Altogether this shows that the sub-problems $(j,C,\k,\{f_{C'}\}_{C'})$ and $(j,C^{(k)},1,\{g_{C''}^{*}\}_{C''})$ have the same set of solutions. 
        In particular, the set of optimal solutions for the sub-problems are the same.
		
		Now we prove the lemma for the case $k<K$. 
        We first show that the DP-cell $(j,C,\k+1,\{f_{C'}^{*}\}_{C'})$ exists. 
        Again we only need to check the last property of the definition of the DP-cells. 
        Consider a job $j'$ and a cell $C'$. 
        From the definition of the respective sets $A$ it follows that $A(j, C, \k)= A(j, C^{(k)}, 1) \dot{\cup}A(j, C, k+1)$. 
        A rectangle $R \in \R(j', C')$ is contained in $A(j, C, \k+1)$ if and only if it is contained in $A(j, C, \k)$ and it does not intersect $A(j, C^{(k)}, 1)$.
		We already know that the DP-cells $(j, C, \k, \{f_{C'}\}_{C'})$ and $(j,C^{(k)},1,\{g_{C''}^{*}\}_{C''})$ exist, so they fulfill the last property in the definition of a DP-cell.
        It follows that either 
        \begin{itemize}
            \item all rectangles in $\R(j', C')$ are contained in $A(j, C, \k)$ and none of them intersect $A(j, C^{(k)}, 1)$, or
            \item all rectangles in $\R(j', C')$ are contained in $A(j, C, \k)$ and
 all rectangles in $\R(j', C')$ are contained in $A(j, C^{(k)}, 1)$, or
            \item no rectangle from $\R(j', C')$ intersects $A(j, C, \k)$.
            \end{itemize}
        In the first case every rectangle in $\R(j', C')$ is contained in $A(j, C, \k+1)$ and in the second and third case no rectangle in $\R(j', C')$ intersects $A(j, C, \k+1)$. 
        Thus, in all cases either all rectangles in $\R(j', C')$ are contained in $A(j, C, \k+1)$ or no rectangle in $\R(j', C')$ intersects $A(j, C, \k+1)$. 
        If all rectangles in $\R(j', C')$ are contained in $A(j, C, \k+1)$ then all those rectangles are also contained in $A(j, C, \k)$ which then implies that $C'=C$ or $C'$ is a descendant of $C$. 
        This shows the last property of the definition of the DP-cells and, therefore the DP-cell $(j,C,\k+1,\{f_{C'}^{*}\}_{C'})$ exists.
		
		Now we show that $\R':=\R^{*} \cap \R_{\down}(j, C^{(k)}, 1)$ is a feasible solution for $(j,C^{(k)},1,\{g_{C''}^{*}\}_{C''})$. 
       It is immediate that $\R' \subseteq \R_{\down}(j, C^{(k)}, 1)$.
        Consider a job $j'$ and a cell $C'$. Then $\R^{*} \cap \R(j', C')$ forms a prefix of $\R(j', C')$ and as $\R_{\down}(j, C^{(k)}, 1)$ either contains all rectangles from $\R(j', C')$ or none of them, also $\R' \cap \R(j', C')=R^{*} \cap \R_{\down}(j, C^{(k)}, 1) \cap \R^{*}$ forms a prefix of $\R(j', C')$.	
		Let $I$ be an interval with
        $L(I) \in \L$ and 
        $L(I) \subseteq A(j, C^{(k)},1)$. 
        Then $L(I) \subseteq A(j, C, \k)$ and $p(\R^{*} \cap \R(I)) \geq d(I)$. 
        Then $L(I) \subseteq A(j, C^{(k)}, 1)$ implies that $\R(I) \subseteq \R_{\down}(j, C^{(k)}, 1)$ and therefore $p(R' \cap \R(I)) = p(\R^{*} \cap \R(I))\geq d(I)$.	
		Let $C'' \in \C(C^{(k)}, 1)$ and $t \in C''$. Then there exists a unique $C' \in \C(C, k)$ with $C'' \subseteq C'$ and therefore $t \in C'$. 
        As $\R^{*}$ is feasible for $(j,C,\k,\{f_{C'}\}_{C'})$ it follows that $p(\R^{*} \cap \R([r_j, t])) \geq d([r_j, t]) +f_{C'}$. As $L([r_j, t]) \subseteq A(j, C^{(k)}, 1)$ we get $\R([r_j, t]) \subseteq \R_{\down}(j, C^{(k)}, 1)$ and therefore $p(\R' \cap \R([r_j, t]))=p(\R^{*} \cap \R([r_j, t])) \geq d([r_j, t]) +f_{C'}=d([r_j, t])+g_{C''}^{*}.$ 
        Thus, $\R'$ is a feasible solution for $(j,C^{(k)},1,\{g_{C''}^{*}\}_{C''})$.
		
		Now we show that $\R''=\R^{*} \cap \R_{\down}(j, C, k+1)$ is a feasible solution for $(j,C,\k+1,\{f_{C'}^{*}\}_{C'})$. 
        Clearly $\R'' \subseteq \R_{\down}(j, C, \k+1)$ and for a job $j'$ and a cell $C'$ the set $\R'' \cap \R(j', C')$ forms a prefix of $\R(j', C')$ because $R^{*} \cap \R(j', C')$ forms a prefix of $\R(j', C')$ and $A(j, C, \k+1)$ either contains all rectangles in $\R(j', C')$ or none of them. 
        For a ray $L(I)$ for an interval $I$ with $L(I) \subseteq A(j, C, \k+1)$ holds $L(I) \subseteq A(j, C, \k)$. 
        So $p(R^{*} \cap \R(I)) \geq d(I)$. 
        As $L(I) \subseteq A(j, C, \k+1)$ it follows that $\R(I) \subseteq \R_{\down}(j, C, k+1)$ and therefore $p(R'' \cap \R(I))=p(\R^{*} \cap \R(I)) \geq d(I)$. 
        So the demand of $I$ is covered by $\R''$. 
        Let $C' \in \C(C, \k+1)$ and consider a ray $L([r_j, t])$ with $t \in C'$. 
        Then $C' \in \C(C, \k)$ implies that $p(R^{*} \cap \R([r_j, t])) \geq d([r_j, t])+f_{C'}$. 
        As $L([r_j, t]) \subseteq A(j, C, \k+1)$ it follows that $\R([r_j, t]) \subseteq \R_{\down}(j, C, k+1)$ and therefore $p(R'' \cap \R([r_j, t]))=p(\R^{*} \cap \R([r_j, t])) \geq d([r_j, t])+f_{C'}$. 
        Thus, $\R''$ is a feasible solution for $(j,C,\k+1,\{f_{C'}^{*}\}_{C'})$.
		
		Now let $\R_1$ be an optimal solution for $(j,C^{(k)},1,\{g_{C''}^{*}\}_{C''})$ and let $\R_2$ be an optimal solution for $(j,C,\k+1,\{f_{C'}^{*}\}_{C'})$. 
		We are going to prove that $\R_1 \cup \R_2$ is a feasible solution for $(j,C,\k,\{f_{C'}\}_{C'})$.	
		It is clear that $\R_1 \cup \R_2 \subseteq \R_{\down}(j, C, \k)$ and as for a job $j'$ and a cell $C'$ the sets $\R_1 \cap \R(j', C')$ and $\R_2 \cap \R(j', C')$ both form a possibly empty prefix of $\R(j', C')$, also the set $\left(\R_1 \cup \R_2 \right)\cap \R(j', C')$ forms a prefix of $\R(j', C')$. 
		Now consider a ray $L(I)\in \L$ for an interval $I=[s, t]$ with $L(I) \subseteq A(j, C, k)$. 
        If $t \in C^{(k)}$, then $L(I) \subseteq A(j, C^{(k)}, 1)$ and $L(I)$ is covered by $\R_1$. 
        Otherwise $L(I) \subseteq A(j, C, \k+1)$ and $L(I)$ is covered by $\R_2$. 
        In both cases $L(I)$ is covered by $\R_1 \cup \R_2$.
		Let $C' \in \C(C, k)$ and consider a ray $L([r_j, t])$ with $t \in C'$. 
        First suppose that $C' \in \C(C, \k+1)$. Then $p(\R_2 \cap \R([r_j, t])) \geq d([r_j, t])+f_{C'}^{*}$ and thus also $p(\left( \R_1 \cup \R_2 \right) \cap \R([r_j, t])) \geq d([r_j, t])+f_{C'}^{*}= d([r_j, t])+f_{C'}$. 
        Secondly, consider $C' \in \C(C, k) \setminus  \C(C, k+1)$. Then $t \in C^{(k)}$ and thus there exists $C'' \in \C(C^{(k)}, 1)$ with $t \in C''$. Then $p(\R_1 \cap \R([r_j, t])) \geq d([r_j, t])+g_{C''}^{*}$ and thus also $p(\left( \R_1 \cup \R_2 \right) \cap \R([r_j, t])) \geq d([r_j, t])+g_{C''}^{*}=d([r_j, t])+f_{C'}$. 
        Consequently, we have in both cases that $p(\left( \R_1 \cup \R_2 \right) \cap \R([r_j, t])) \geq d([r_j, t])+f_{C'}$. So $\R_1 \cup \R_2$ is a feasible solution for $(j,C,\k,\{f_{C'}\}_{C'})$.
		
		As $c(\R_1 \cup \R_2)=c(\R_1)+c(\R_2) \leq c(\R')+c(\R'')=c(\R^{*})$ and $\R^{*}$ is an optimal solution for $(j,C,\k,\{f_{C'}\}_{C'})$ it follows that $\R_1 \cup \R_2$ is also an optimal solution for $(j,C,\k,\{f_{C'}\}_{C'})$ which concludes the proof of the lemma.
	\end{proof}

Finally, suppose that $(j,C,\k,\{f_{C'}\}_{C'})$ is a non-canonical
DP-cell and that $\ell(C)=\ell_{\max}$. We define $\left\{ f_{C'}^{*}\right\} _{C'}$
to be the restriction of $\left\{ f_{C'}\right\} _{C'}$ to the intervals
in $\C(C,\k+1)$ and reduce the problem to the DP-cell $\OPT(j,C,\k+1,\{f_{C'}^{*}\}_{C'})$.
This is justified by the following lemma.
\begin{lem}
\label{lem:DP-non-canonical-ell-max}Assume that $(j,C,\k,\{f_{C'}\}_{C'})$
is a non-canonical DP-cell%
 and that $\ell(C)=\ell_{\max}$. Then $k<K$,
$(j,C,\k+1,\{f_{C'}^{*}\}_{C'})$ is a DP-cell and $\OPT(j,C,\k+1,\{f_{C'}^{*}\}_{C'})$
is an optimal solution for $(j,C,\k,\{f_{C'}\}_{C'})$. 
\end{lem}
\begin{proof}
	From the definition of $\C(C, k)$ we conclude that the intervals in $\C(C, k)$ are the intervals $[t, t+1)$ with $t \in \N$ and $\beg(C) \leq t \leq \en(C)-1$. 
    We assumed that $A(j,C,\k)$ contains at least one rectangle. It can only contain rectangles in sets $\R(j')$ with $j' \geq j$. 
    As the jobs are ordered by their release times and the job release times are pairwise distinct we know that $r_j <r_{j'}$. 
    So the leftmost $x$-coordinate of each rectangle belonging to some job $j' \geq j$ is at least $r_j$. 
    Thus $A(j,C,\k)$ can only contain rectangles whose leftmost $x$-coordinate is at least $r_j$. 
    As $A(j, C, \k)$ contains a rectangle, we obtain $r_j \leq \en(C)-1$. So $A(j, C, \k)$ contains a rectangle in $\R(j,C)$.
    Due to the the last property in the definition of the DP-cell $(j,C,\k,\{f_{C'}\}_{C'})$
    all rectangles in $\R(j, C)$ are contained in $A(j, C, \k)$. 
  Also, since the values $\{f_{C'}^{*}\}_{C'}$ are taken from the values $\{f_{C'}\}_{C'}$, the upper bounds for them are automatically respected.

	The segments in $\mathrm{Seg}(j, C)$ are the intervals $[t, t+1)$ with $t \in \N$ and $r_j \leq t \leq \en(C)-1$. 
    As all rectangles in $\R(j, C)$ are contained in $A(j, C, \k)$ we obtain $\beg(C)+k-1 \leq r_j$. As $\beg(C)+k-1 = r_j$ would imply that the DP-cell $(j,C,\k,\{f_{C'}\}_{C'})$ is canonical which we assumed to be not the case, this implies $\beg(C)+k \leq r_j$. 
    As we already showed $r_j \leq \en(C)-1$, this yields $\beg(C)+k\leq \en(C)-1$ and thus $k < \en(C)-\beg(C)$. This implies that $A(j, C, k+1)\neq \emptyset$ and $k < \en(C)-\beg(C)\leq K$.
		 
	We know that every rectangle contained in $A(j, C, \k)$ has a leftmost $x$-coordinate of at least $r_j$ and as $\beg(C)+k \leq r_j$ it is also contained in $A(j, C, \k+1)$. 
    This yields $\R_{\down}(j, C, \k)=\R_{\down}(j, C, \k+1)$. As the DP-cell $(j,C,\k,\{f_{C'}\}_{C'})$ exists, this implies the existence of the DP-cell $(j,C,\k+1,\{f_{C'}^{*}\}_{C'})$. 
		 
	In the remaining part we show that the sub-problems $(j,C,\k,\{f_{C'}\}_{C'})$ and $(j,C,\k+1,\{f_{C'}^{*}\}_{C'})$ have the same set of solutions. 
    We already proved $\R_{\down}(j, C, \k)=\R_{\down}(j, C, \k+1)$. So consider a ray $L(I)$ for an interval $I$ with $L(I) \subseteq A(j, C, \k)$. 
    If $L(I) \subseteq A(j, C, \k) \setminus A(j, C, \k+1)$ then it does not intersect any rectangles so we do not need to check it by Lemma~\ref{lem:empty-ray}. 
    And if $L(I) \subseteq A(j, C, \k+1)$ then the demand for this ray is the same for both sub-problems. 
    Consider an interval $C' \in \C(C, \k)$ and let $t \in C'$ such that there is a ray $L([r_j, t])$. Then $t \geq r_j$ and thus $C' \in \C(C, \k+1)$. 
    So the set of rays with additional demand (of $f_C'$ etc.) are also the same. Altogether, the sub-problems $(j,C,\k,\{f_{C'}\}_{C'})$ and $(j,C,\k+1,\{f_{C'}^{*}\}_{C'})$ have indeed the same set of solutions. 
    This also implies that they have the same optimal solutions.
\end{proof}

Finally, we output the solution in the cell
$(1,C_0,1,\{f^0_{C'}\}_{C'})$ with $f^0_{C'}=0$ for each $C'$.
Our dynamic program yields the following lemma.
\begin{lem}
We can compute an optimal solution to an instance of our geometric
covering problem in time $(nP)^{O(K^2)} \le (nP)^{O_{\epsilon(1)}}$.
\end{lem}

\begin{proof}
The number of cells in the hierachical decomposition is bounded by $O(K^2 nP)$.
For each DP-cell, there are at most $(nP)^{K^2}$ many values for $\{f_{C'}\}_{C'}$ to be considered.
Thus, the number of DP-cells $(j, C, k, \{f_{C'}\}_{C'})$ is at most $O(n \cdot K^2 nP \cdot K \cdot (n P)^{K^2}) \le (nP)^{O(K^2)}$.
We fill in the DP-cells starting with those of the base case; then we compute the entries of the DP-cells in decreasing lexicographic order of $(j, \ell(C), k)$, following the transition  rules of Lemmas~\ref{lem:DP-canonical}-\ref{lem:DP-non-canonical-ell-max}
and guess $\mathcal{R}^{*}(j,C)$ (for which there are at most $K^2=O_{\epsilon}(1)$ options) when we apply Lemma~\ref{lem:DP-canonical}. Note that in some cases, our constructed (candidate) solution is infeasible for the respective given DP-cell. In this case we store that the given DP-cell does not have a feasible solution.
Guessing the necessary values and applying the rules above can be done in a running time which is polynomial in 
in $n$, $P$, and $K$. Thus, our overall running time is $(nP)^{O(K^2)}$.
\end{proof}
Together with Lemma~\ref{lem:reduction} this yields our main theorem.

\begin{theorem} There is a polynomial time $(1+\epsilon)$-approximation
algorithm for weighted flow time on a single machine when preemptions
are allowed. \end{theorem}

\bibliographystyle{plain}
\bibliography{references}

\end{document}